\newif\ifarxiv
\arxivtrue

\ifarxiv
\documentclass[letterpaper,11pt]{article}
\usepackage{fullpage}
\usepackage{amsfonts,amsmath,amssymb,amsthm}
\usepackage{cite,comment,enumerate,hyperref}
\newtheorem{theorem}{Theorem}

\newtheorem{lemma}[theorem]{Lemma}

\theoremstyle{definition}
\newtheorem{definition}{Definition}
\theoremstyle{remark}	

\usepackage{xcolor,xspace}
\else
\documentclass[sigconf]{acmart}
\fi

\usepackage{algorithm}
\usepackage{algorithmic}

\begin{document}
\title{Efficient and Simple Algorithms for Fault-Tolerant Spanners}
\ifarxiv
\author{Michael Dinitz\thanks{Supported in part by NSF award CCF-1909111}\\Johns Hopkins University \and Caleb Robelle\\University of Maryland, Baltimore County}
\date{}
\fi

\ifarxiv
\maketitle
\fi
\begin{abstract}
    It was recently shown that a version of the greedy algorithm gives a construction of fault-tolerant spanners that is size-optimal, at least for vertex faults.  However, the algorithm to construct this spanner is not polynomial-time, and the best-known polynomial time algorithm is significantly suboptimal.  Designing a polynomial-time algorithm to construct (near-)optimal fault-tolerant spanners was given as an explicit open problem in the two most recent papers on fault-tolerant spanners ([Bodwin, Dinitz, Parter, Vassilevka Williams SODA '18] and [Bodwin, Patel PODC '19]).  We give a surprisingly simple algorithm which runs in polynomial time and constructs fault-tolerant spanners that are extremely close to optimal (off by only a linear factor in the stretch) by modifying the greedy algorithm to run in polynomial time.  To complement this result, we also give simple distributed constructions in both the LOCAL and CONGEST models.
\end{abstract}

\section{Introduction}

Let $G = (V, E)$ be a graph, possibly with edge lengths $w : E \rightarrow \mathbb{R}_{\geq 0}$. A $t$-spanner of $G$, for $t \geq 1$, is a subgraph $G' = (V, E')$ that preserves all pairwise distances within factor $t$, i.e.,
\begin{equation} \label{eq:stretch}
d_{G'}(u,v) \leq t \cdot d_G(u,v)
\end{equation}
for all $u,v \in V$ (where $d_H$ denotes the shortest-path distance in a graph $H$).  The distance preservation factor $t$ is called the \emph{stretch} of the spanner.  Less formally, graph spanners are a form of sparsifiers that approximately preserve distances (as opposed to other notions of graph sparsification which approximately preserve cuts~\cite{BenczurK:15}, the spectrum~\cite{SpielmanS11,BatsonSS14}, or other graph properties). When considering spanners through the lens of sparsification, perhaps the most important goal in the study of graph spanners is understanding the tradeoff between the stretch and the sparsity.  The main result in this area, which is tight assuming the ``Erd\H{o}s girth conjecture''~\cite{erdHos1964extremal}, was given by Alth\"ofer et al.:
\begin{theorem}[\cite{AlthoferDDJS:93}] \label{thm:ADDJS}
For every positive integer $k$, every weighted graph $G = (V, E)$ has a $(2k-1)$-spanner with at most $O(n^{1+1/k})$ edges.
\end{theorem}

This notion of graph spanners was first introduced by Peleg and Sch\"affer~\cite{PS89} and Peleg and Ullman~\cite{PU89} in the context of distributed computing, and has been studied extensively for the last three decades in the distributed computing community as well as more broadly.  Spanners are not only inherently interesting mathematical objects, but they also have an enormous number of applications.  A small sampling includes uses in distance oracles~\cite{ThorupZ:05}, property testing~\cite{BGJRW09,BBGRWY14}, synchronizers~\cite{PU89}, compact routing~\cite{ThorupZ:01}, preprocessing for approximation algorithms~\cite{BKM09,DKN17}), and many others.  

Many of these applications, particularly in distributed computing, arise from modeling computer networks or distributed systems as graphs. But one aspect of distributed systems that is not captured by the above spanner definition is the possibility of \emph{failures}. We would like our spanner to be robust to failures, so that even if some nodes fail we still have a spanner of what remains. More formally, $G'$ is an $f$-(vertex-)fault-tolerant $t$-spanner of $G$ if for every set $F \subseteq V$ with $|F| \leq f$ the spanner condition holds for $G \setminus F$, i.e., 
\[
d_{G' \setminus F}(u,v) \leq t \cdot d_{G \setminus F}(u,v)
\]
for all $u,v \in V \setminus F$.  If $F$ is instead an edge set then this gives a definition of an $f$-edge-fault-tolerant $t$-spanner.

This notion of fault-tolerant spanners was first introduced by Levcopoulos, Narasimhan, and Smid~\cite{LNS98} in the context of geometric spanners (the special case when the vertices are in Euclidean space and the distance between two points is the Euclidean distance), and has since been studied extensively in that setting~\cite{LNS98,lukovszki1999new,czumaj2004fault,NS07}.  Note that in the geometric setting $d_{G \setminus F}(u,v) = d_G(u,v)$ for all $u,v \in V \setminus F$, since faults do not change the underlying geometric distances.  

In general graphs, though, $d_{G \setminus F}(u,v)$ may be extremely different from $d_G(u,v)$, making this definition more difficult to work with.  The first results on fault-tolerant graph spanners were by Chechik, Langberg, Peleg, and Roditty~\cite{ChechikLPR:10}, who showed how to modify the Thorup-Zwick spanner~\cite{ThorupZ:05} to be $f$-fault-tolerant with an additional cost of approximately $k^f$: the number of edges in the $f$-fault-tolerant $(2k-1)$-spanner that they create is approximately $\tilde O(k^f n^{1+1/k})$ (where $\tilde O$ hides polylogarithmic factors).  Since~\cite{ChechikLPR:10} there has been a significant amount of work on improving the sparsity, particularly as a function of the number of faults $f$ (since we would like to protect against large numbers of faults but usually care most about small stretch values).  First, Dinitz and Krauthgamer~\cite{DinitzK:11} improved the size to $\tilde O(f^{2-1/k} n^{1+1/k})$ by giving a black-box reduction to the traditional non-fault-tolerant setting.  Then Bodwin, Dinitz, Parter, and Vassilevska Williams~\cite{BDPW18} decreased this to $O(\exp(k) f^{1-1/k} n^{1+1/k})$, which they also showed was optimal (for vertex faults) as a function of $f$ and $n$ (i.e., the only non-optimal dependence was the $\exp(k)$).  Unlike previous fault-tolerant spanner constructions, this optimal construction was based off of a natural greedy algorithm (the natural generalization of the greedy algorithm of~\cite{AlthoferDDJS:93}).  An improved analysis of the same greedy algorithm was then given by Bodwin and Patel~\cite{BP19}, who managed to show the fully optimal bound of $O(f^{1-1/k} n^{1+1/k})$.

Unlike the previous fault-tolerant spanner construction of~\cite{DinitzK:11} and the greedy non-fault-tolerant algorithm of~\cite{AlthoferDDJS:93}, the greedy algorithm of~\cite{BDPW18,BP19} has a significant weakness: it takes exponential time.  Obtaining the same (or similar) size bound in \emph{polynomial} time was explicitly mentioned as an important open question in both~\cite{BDPW18} and~\cite{BP19}.    

\subsection{Our Results and Techniques}
In this paper we design a surprisingly simple algorithm to construct nearly-optimal fault-tolerant spanners in polynomial time, in both unweighted and weighted graphs.
\begin{theorem} \label{thm:main}
There is a polynomial time algorithm which, given integers $k \geq 1$ and $f \geq 1$ and a (weighted) graph $G = (V, E)$ with $|V| = n$ and $|E|=m$, constructs an $f$-fault-tolerant $(2k-1)$-spanner with at most $O\left(k f^{1-1/k} n^{1+1/k}\right)$ edges in time $O(m kf^{2-1/k} n^{1+1/k})$.
\end{theorem}

Note that while we are a factor of $k$ away from complete optimality (for vertex faults), this is truly optimal when the stretch is constant and, for non-constant stretch values, is still significantly sparser than the analysis of the exponential time algorithm by~\cite{BDPW18} (which lost an exponential factor in $k$).

The main idea in our algorithm is to replace the exponential-time subroutine used in the greedy algorithm of~\cite{BDPW18,BP19} with an appropriate polynomial-time approximation algorithm.  More specifically, the main step of the exponential time greedy algorithm is to consider whether a given candidate edge is ``already spanned'' by the subgraph $H$ that has already been built.  This means determining whether, for some candidate edge $\{u,v\}$, there is a fault set $F$ with $|F| \leq f$ such that $d_{H \setminus F}(u,v) > (2k-1) \cdot d_{G \setminus F}(u,v)$.  If such a fault set exists then the algorithm adds $\{u,v\}$ to $H$, and otherwise does not\footnote{Note that in the fault-free case this just means checking whether there is already a path of stretch at most $(2k-1)$ between the endpoints, which is precisely the original greedy algorithm of~\cite{AlthoferDDJS:93}.}.  In both~\cite{BDPW18} and~\cite{BP19}, the only method given to find such a set $F$ was to try all possible sets, giving running time that is exponential in $f$ and thus exponential in the size of the input.

Our main approach is to speed this up by designing a polynomial-time algorithm to replace this exponential-time step.  Unfortunately, the corresponding problem (known as \textsc{Length-Bounded Cut}) is NP-hard~\cite{BEHKSS06}, so we cannot hope to actually solve it efficiently.  Instead, we design an \emph{approximation algorithm} for \textsc{Length-Bounded Cut} and use it instead.  We end up with a fairly weak approximation (basically a $k$-approximation), and one which only holds in the unweighted case.  But this turns out to be enough for the unweighted case: it intuitively allows us to build (in polynomial time) an $f$-fault-tolerant spanner with the size of a $kf$-fault-tolerant spanner, which changes the size from $O(f^{1-1/k}n^{1+1/k})$ to $O((kf)^{1-1/k} n^{1+1/k}) = O(k f^{1-1/k} n^{1+1/k})$.  However, this is only intuition.  The graph we end up creating is not necessarily even a subgraph of the $kf$-fault-tolerant spanner that the true greedy algorithm would have built, so we cannot simply argue that our algorithm returns something with at most as many edges as the greedy $kf$-fault-tolerant greedy spanner.  Instead, we need to analyze the size of our spanner from scratch.  Fortunately, we can do this by simply following the proof strategy of~\cite{BP19} with only some minor modifications.   

A natural approach to the weighted case would be to try to generalize this by creating an $O(k)$-approximation for \textsc{Length-Bounded Cut} in the weighted setting.  Such an algorithm would certainly suffice, but unfortunately we do not know how to design any nontrivial approximation algorithm for \textsc{Length-Bounded Cut} in the presence of weights.  While this might appear to rule out using a similar technique, we show that special properties of the greedy algorithm allow us to essentially reduce to the unweighted setting.  We use the weights to determine the order in which we consider edges, but for the rest of the algorithm we simply ``pretend'' to be in the unweighted setting.  Since the size bound for the unweighted case worked for any ordering, that same size bound will apply to our spanner.  And then we can use the fact that we considered edges in order of nondecreasing weights to argue that the subgraph we create is in fact an $f$-fault-tolerant $(2k-1)$-spanner even though we ignored the weights.

\subsubsection*{Distributed Settings}
While the focus of this paper is on a centralized polynomial-time algorithm since the existence of such an algorithm was an explicit open question from~\cite{BDPW18} and \cite{BP19}, we complement this result with some simple algorithms in the standard LOCAL and CONGEST models of distributed computation.  

In the LOCAL model, we can use standard network decompositions to find a clustering of the graph where the clusters have low diameter, every edge is in at least one cluster, and the clustering comes from $O(\log n)$ partitions.  Since in the LOCAL model we are allowed unbounded message sizes, this means that in $O(\log n)$ time we can send the subgraph induced by each cluster to the cluster center (an arbitrary node in the cluster), who can then locally run the greedy algorithm on that cluster and then inform the nodes in the cluster about the edges that have been chosen.  This will take only $O(\log n)$ communication rounds (since clusters have diameter $O(\log n)$) and will incur only an extra $O(\log n)$ factor in the number of edges (since the clustering can be divided into $O(\log n)$ partitions).  

In the CONGEST model we cannot apply this approach (even though we could find a similar clustering) because we are not able to gather large induced subgraphs at the cluster centers (due to the bound on message sizes).  Instead, we show that the older fault-tolerant spanner construction of~\cite{DinitzK:11} can be combined with the standard (non-fault-tolerant) spanner algorithm in the CONGEST model due to Baswana and Sen~\cite{BaswanaS:07} to give a fault-tolerant spanner algorithm in CONGEST.  This approach means that the size increases to $O(kf^{2-1/k} n^{1+1/k} \log n)$ (so we are a factor of $f \log n$ away from the bounds of the polynomial-time greedy algorithm), but the number of rounds needed is quite small despite the limitation on message sizes ($O(f^2(\log f + \log \log n) + k^2 f \log n)$ rounds).

\section{Notation and Preliminaries}
We will be discussing graphs $G = (V, E)$ where $n = |V|$ and $m = |E|$.  Sometimes these graphs will also have a weight function $w : E \rightarrow \mathbb{R}_{\geq 0}$.  We will slightly abuse notation to let $w(u,v) = w(\{u,v\})$ for all $\{u,v\} \in E$.  For a (possibly weighted) graph $G$, we will let $d_G(u,v)$ denote the length of the shortest (lowest-weight) path from $u$ to $v$ (if no such path exists then this length is $\infty$).  For any $C \subseteq V$, we let $G[C]$ denote the subgraph of $G$ induced by $C$.  For $F \subseteq V$ let $G \setminus F$ be $G[V \setminus F]$, and for $F \subseteq E$ let $G \setminus F$ be $(V, E \setminus F)$.  

\begin{definition} \label{def:FT}
Let $G = (V, E)$ be a (possibly weighted) graph.  A subgraph $H$ of $G$ is an $f$-vertex-fault-tolerant ($f$-VFT) $t$-spanner of $G$ if $d_{H \setminus F}(u,v) \leq t \cdot d_{G \setminus F}(u,v)$ for all $F \subseteq V$ with $|F| \leq f$ and $u,v \not\in F$.  A subgraph $H$ of $G$ is an $f$-edge-fault-tolerant ($f$-EFT) $t$-spanner of $G$ if $d_{H \setminus F}(u,v) \leq t \cdot d_{G \setminus F}(u,v)$ for all $F \subseteq E$ with $|F| \leq f$.
\end{definition}

Throughout this paper, for simplicity we will only discuss the vertex fault-tolerant case since that is the more difficult one to prove upper bounds for.  The proofs for the edge fault-tolerant case are essentially identical.  

We first show an equivalent definition that will let us restrict which pairs of vertices we care about.

\begin{lemma} \label{lem:simple}
Let $G = (V, E)$ be a graph with weight function $w$ and let $H$ be a subgraph of $G$.  Then $H$ is an $f$-VFT $t$-spanner of $G$ if and only if $d_{H \setminus F}(u,v) \leq t \cdot w(u,v)$ for all $F \subseteq V$ with $|F| \leq f$ and $u,v \in V \setminus F$ such that $\{u,v\} \in E$ and $d_{G \setminus F}(u,v) = w(u,v)$
\end{lemma}
\begin{proof}
The only if direction is immediately implied by Definition~\ref{def:FT}, since for any $F \subseteq V$ with $|F| \leq f$ and $u,v \in V \setminus F$ such that $\{u,v\} \in E$ and $d_{G \setminus F}(u,v) = w(u,v)$, we know from Definition~\ref{def:FT} that $d_{H \setminus F}(u,v) \leq t \cdot d_{G \setminus F}(u,v) \leq t \cdot w(u,v)$.  

For the if direction, let $F \subseteq V$ with $|F| \leq f$ and $u,v \in V \setminus F$.  Let $P = (u = x_0, x_1, \dots, x_p = v)$ be the shortest path in $G \setminus F$ between $u$ and $v$.  If $p = 1$ then $P = (u,v)$, and thus $d_{H \setminus F}(u,v) = w(u,v) = d_{G \setminus F}(u,v)$.  If $p > 1$, then we know that $d_{G \setminus F}(x_{i-1}, x_i) = w(x_{i-1}, x_{i})$ for all $i \in \{1,2,\dots, p\}$, and thus
\begin{align*}
d_{H \setminus F}(u,v) &\leq \sum_{i=1}^p d_{H \setminus F}(x_{i-1}, x_i) \leq \sum_{i=1}^p t \cdot w(x_{i-1}, x_i)\\
&= t \sum_{i=1}^p w(x_{i-1}, x_i) = t \cdot d_{G \setminus F}(u,v).
\end{align*}
Hence $H$ is an $f$-VFT $t$-spanner of $G$.
\end{proof}

The original greedy algorithm for fault-tolerant spanners was introduced and analyzed by~\cite{BDPW18}, with an improved analysis by~\cite{BP19}, and is given in Algorithm~\ref{ALG:old}.  The part of this algorithm which takes exponential time is the ``if'' condition, i.e., checking whether there is a fault set which hits all stretch-$(2k-1)$ paths.  For edge fault-tolerance, the algorithm is the same except that $F$ is an edge set.
\begin{algorithm}
\caption{Greedy $f$-VFT $(2k-1)$-Spanner Algorithm}
\label{ALG:old}
\begin{algorithmic}
\STATE $\mathbf{function}$ FT-GREEDY$(G = (V,E,w),k,f)$\\
\STATE $H\leftarrow (V,\emptyset,w)$
\FORALL{$\{u,v\}\in E$ in nondecreasing weight order} 
    \IF{there exists a set $F$ of at most $f$ vertices such that $d_{H\setminus F}(u,v) > (2k-1) w(u,v)$}
        \item add $\{u,v\}$ to H
    \ENDIF
\ENDFOR
\RETURN H
\end{algorithmic}
\end{algorithm}


\section{Unweighted Graphs}

In this section we design a polynomial-time algorithm for the special case of unweighted (or unit-weighted) graphs.  We begin by designing a simple approximation algorithm for the \textsc{Length-Bounded Cut} problem, and then show that this algorithm can be plugged into the greedy algorithm with only a small loss.

\subsection{Length-Bounded Cut}
In order to design a polynomial-time variant of the greedy algorithm, we want to replace the ``if'' condition by something that can be computed in polynomial time.  While there are many possibilities, there are two obvious approaches: we could try to compute the maximum $t$ such that there is a fault set of size $f$ which hits all $t$-hop paths, or we could try to compute the minimum $f$ such that there is a fault set of size $f$ which hits all $t$-hop paths.  It turns out that this second approach is more fruitful.

Consider the following problem, known as the \textsc{Length-Bounded Cut} problem~\cite{BEHKSS06}.  The input is an unweighted graph $G = (V, E)$ with $|V|=n$ and $|E| = m$, vertices $u,v \in V$ (known as the \emph{terminals}), and a positive integer $t$.  A \emph{length-$t$-cut} is a subset $F \subseteq V \setminus \{u,v\}$ such that $d_{G \setminus F}(u,v) > t$.  The goal is to find the length-$t$-cut of minimum cardinality.   

We are essentially going to design a $t$-approximation for this problem.  But since we do not need the full power of this approximation, in order to speed it up we will instead consider a gap decision version of the problem.  In the LBC($t, \alpha$) problem, the input is the same as in \textsc{Length-Bounded Cut} but there is an additional input parameter $\alpha$.  If there is a length-$t$-cut of size at most $\alpha$, then we must return YES.  If there is no length-$t$-cut of size at most $\alpha t$, then we must return NO.  For intermediate values we are allowed to return either YES or NO.

Recall that breadth-first search (BFS) finds shortest paths in unweighted graphs in $O(m+n)$ time.  So we can use BFS to check whether there is a path with at most $t$ hops from $u$ to $v$ in $O(m + n)$ time. This gives the following natural algorithm (Algorithm~\ref{ALG:LBC}), which is essentially the standard ``frequency'' approximation of \textsc{Set Cover} (or \textsc{Hitting Set}).  

\begin{algorithm}
\caption{Algorithm for LBC($t, \alpha$)}
\label{ALG:LBC}
\begin{algorithmic}
\STATE $F \leftarrow \emptyset$
\FOR{$i = 1$ to $\alpha+1$}
    \STATE Run BFS to find a path $P$ of length at most $t$ from $u$ to $v$ in $G \setminus F$ if one exists.
    \IF{no such $P$ exists} \RETURN YES
    \ELSE \STATE Add all vertices of $P \setminus \{u,v\}$ to $F$
    \ENDIF 
\ENDFOR
\RETURN NO
\end{algorithmic}
\end{algorithm}

\begin{theorem} \label{thm:LBC}
Algorithm~\ref{ALG:LBC} correctly decides LBC($t, \alpha$) and runs in $O((m + n)\alpha)$ time.
\end{theorem}
\begin{proof}
By the running time of BFS, we know that each iteration of Algorithm~\ref{ALG:LBC} takes $O(m + n)$ time, and thus the total time is $O((m + n)\alpha)$ as claimed.

Suppose that there is a length-$t$-cut $F^*$ of size at most $\alpha$.  Then for every path $P$ which our algorithm considers (and adds to $F$), it must be the case that $|P \cap F^*| \geq 1$ since $F^*$ must hit all paths of length at most $t$.  Since we remove each path we consider (by adding it to $F$), this means that there will be no more such paths after at most $\alpha$ iterations and thus the algorithm will return YES as required.

Now suppose that every length-$t$-cut has size larger than $\alpha t$.  Since we add at most $t$ vertices to $F$ in each iteration, at the beginning of iteration $\alpha+1$ the set $F$ has size at most $\alpha t$.  Thus in every iteration some path $P$ of length at most $t$ exists, so the algorithm will return NO.  
\end{proof}

To handle edge fault-tolerance, we need to slightly change the definition of LBC($t, \alpha$) to be about edge sets rather than vertex sets, so in the algorithm $F$ is an edge set and we add the edges of $P$ rather than the vertices.  But other than that trivial change, the algorithm and analysis are identical.    

\subsection{Modified Greedy}
Let $G=(V,E)$ be an undirected unweighted graph. We will modify Algorithm~\ref{ALG:old} by using our new algorithm for LBC, Algorithm~\ref{ALG:LBC}.  For an EFT spanner algorithm, we simply use the edge-based version of Algorithm~\ref{ALG:LBC}.

\begin{algorithm}
\caption{Modified Greedy VFT Spanner Algorithm}
\label{ALG:new}
\begin{algorithmic}
\STATE $\mathbf{function}$ FT-GREEDY$(G = (V,E),k,f)$\\
\STATE $H\leftarrow (V,\emptyset,w)$
\FORALL{$\{u,v\}\in E$ in arbitrary order} 
    \IF{Algorithm~\ref{ALG:LBC} returns YES when run on input graph $H$ with terminals $u,v$ and $t = 2k-1$ and $\alpha = f$}
    \STATE Add $\{u,v\}$ to $H$
    \ENDIF
\ENDFOR
\RETURN H
\end{algorithmic}
\end{algorithm}

We first prove that this algorithm does indeed return a valid solution, despite the use of an approximation algorithm to determine whether or not to add an edge (we prove this only for VFT for simplicity, but the proof for EFT is analogous).

\begin{theorem} \label{thm:correct-unweighted}
Algorithm~\ref{ALG:new} returns an $f$-VFT $(2k-1)$-spanner.
\end{theorem}
\begin{proof}
Let $F\subseteq V$ be an arbitrary fault set  with $|F|\leq f$ and $\{u,v\} \in E$ with $u,v \not\in F$.  By Lemma~\ref{lem:simple}, we just need to show that $d_{H \setminus F}(u,v) \leq 2k-1$ (since $G$ is unweighted) in order to prove the theorem.  Clearly this is true if $\{u,v\} \in E(H)$.  If $\{u,v\} \not \in E(H)$, then when the algorithm considered $\{u,v\}$ it must have been the case that Algorithm~\ref{ALG:LBC} returned NO.   Theorem~\ref{thm:LBC} then implies that every length-$(2k-1)$-cut on $H$ (for $u,v$) has size larger than $f$.  Thus $F$ is \emph{not} a length-$(2k-1)$-cut in $H$ for $u,v$, and so $d_{H \setminus F}(u,v) \leq 2k-1$.
\end{proof}

Now we want to bound the size of the returned spanner.  To do this, a natural approach would be to argue that the spanner it returns is a subgraph of the greedy $((2k-1)f)$-VFT spanner, since it seems like whenever our modified algorithm requires us to add an edge it has found a cut certifying that the greedy $((2k-1)f)$-VFT spanner would also have had to add that edge.  Unfortunately, this is not true since the modified algorithm might not add some edges that the true greedy algorithm would have added, and thus later on our algorithm might have to actually add some edges that the true greedy algorithm would not have had to add.

The next natural approach would be to try to use the analysis of~\cite{BP19} as a black box.  Unfortunately we cannot do this either, since the lemmas they use are specific to the true greedy algorithm rather than our modification.  However, it is straightforward to modify their analysis so that it continues to hold for our modified algorithm, with only an additional loss of a factor  of $k$.  We do this here for completeness.  As in~\cite{BP19}, we start with the definition of a \emph{blocking set}, and then give two lemmas using this definition.  And also as in~\cite{BDPW18,BP19}, we only prove this for VFT, as the proof for EFT is essentially identical.

\begin{definition}[\cite{BP19}]
For any graph $G=(V,E)$, we define $B\subseteq V\times E$ to be a \emph{$t$-blocking set} of $G$ if for all $(v,e)\in B$, we have $v \not\in e$ and for any cycle $C$ in $G$ with $|C|\leq t$, there exists $(v,e)\in B$ such that $v,e\in C$.
\end{definition}

\begin{lemma} \label{lemma: block set size}
Any graph $H$ returned by Algorithm~\ref{ALG:new} with parameters $k,\,f$ has a $(2k)$-blocking set of size at most $(2k-1)f|E(H)|$.
\end{lemma}

It was shown in \cite{BP19} that the graph $H$ returned by the standard VFT greedy algorithm with parameters $k,\,f$ has a ($2k$)-blocking set of size at most $f|E(H)|$.\footnote{In~\cite{BP19} the parameter ``$k$'' is used to denote the stretch, while for us the stretch is $2k-1$, and thus there are slight constant factor differences between the statements as written in~\cite{BP19} and our interpretation of their statements.  But our statements about~\cite{BP19} are correct under this change of variables.}  So our modified algorithm satisfies the same lemma up to a factor of $O(k)$.  The proof is almost identical in our case; we essentially replace all instances of $f$ in their proof with $(2k-1)f$.

\begin{proof}[Proof of Lemma~\ref{lemma: block set size}]
Let $e = \{u,v\}$ be some edge in $E(H)$, and let $H'$ be the subgraph maintained by the algorithm just before $e$ is added to $E(H)$ (so $H'$ is a subset of the final $H$).  Since $e$ was added by Algorithm~\ref{ALG:new}, when it was considered Algorithm~\ref{ALG:LBC} must have returned YES.  Thus by Theorem~\ref{thm:LBC} there is some set $F_e \subseteq V \setminus \{u,v\}$ with $|F_e| \leq f(2k-1)$ such that $d_{H' \setminus F_e}(u,v) > 2k-1$.  

Now we can define the blocking set: let $B = \{(x,e) : e \in E(H), x \in F_e\}$.

Since $|F_e| \leq f(2k-1)$ for all $e \in E(H)$, we immediately get that $|B| \leq |E(H)| f (2k-1)$ as claimed.  So we now need to show that $B$ is a $(2k)$-blocking set.  To see this, let $C$ be any cycle with at most $2k$ vertices in $H$, and let $e = \{u,v\}$ be the last edge of this cycle to be added to $H$.  Let $H'$ be the subgraph of $H$ built by the algorithm just before $e$ is added.  Then $C \setminus e$ is a $u-v$ path in $H'$ of length at most $2k-1$, and thus there is some $x \in C \setminus \{u,v\}$ that is in $F_e$.  Thus $(x,e) \in B$.
\end{proof}

Now we know that the spanner returned by Algorithm~\ref{ALG:new} has a small blocking set.  The next lemma implies that any such graph must have a dense but high-girth subgraph.

\begin{lemma} \label{lemma: subgraph bound}
Let $H$ be any graph on $n$ nodes and $m$ edges (with $f=o(n)$) that has a $(2k)$-blocking set $B$ of size at most $(2k-1)fm$. Then $H$ has a subgraph on $O(n/(kf))$ nodes and $\Omega(m/(kf)^2)$ edges that has girth greater than $2k$.
\end{lemma}
\begin{proof}
Let $H'$ denote the induced subgraph of $H$ on a uniformly random subset of exactly $\lfloor n/(2(2k-1)f)\rfloor$ nodes.  Let $B':=B\cap(V(H')\times E(H'))$, and let $H''$ denote the graph obtained by removing from $H'$ every edge contained in any pair in $B'$.  The graph $H''$ will be the one we analyze.  

The easiest property to analyze is the number of nodes in $H''$: there are precisely $\lfloor n/(2(2k-1)f)\rfloor$ vertices in $H''$, which is $O(n/(kf))$ as claimed.

The next easiest property of $H''$ to prove is the girth.  Let $C$ be a cycle in $H$ with at most $2k$ nodes.  $C$ is either in $H'$ or it is not.  If it is not in $H'$ then some vertex in $C$ is not in $V(H')$, and thus $C$ is not in $H''$.  On the other hand, if $C$ is in $H'$ then by the definition of $B$ there is some edge $(x,e) \in B$ so that $e \in C$, and also $(x,e) \in B'$, and thus $C$ does not exist in $H''$.

To analyze $|E(H'')|$, we start with the following observations.
\begin{itemize}
    \item Each $\{u,v\}\in E(H)$ remains in $E(H')$ if $u,v\in V(H')$. This happens with probability
        \begin{align*}
        &\quad \frac{\lfloor n/(2(2k-1)f)\rfloor}{n}\cdot\frac{\lfloor n/(2(2k-1)f)\rfloor-1}{n-1}\\
        & \geq(1-o(1))\frac{1}{4((2k-1)f)^2}
        \end{align*}
    \item Each $(x,\{u,v\})\in B$ remains in $B'$ if $u,v,x\in V(H')$. This happens with probability
        \begin{align*}
        &\frac{\lfloor n/(2(2k-1)f)\rfloor}{n}\cdot\frac{\lfloor n/(2(2k-1)f)\rfloor-1}{n-1}\cdot\frac{\lfloor n/(2(2k-1)f)\rfloor-2}{n-2}\\
        &\leq \frac{1}{8((2k-1)f)^3}
        \end{align*}
\end{itemize}

Now we can use these observations to compute the expected size of $E(H'')$:
\begin{align*}
    \mathbb{E}[|E(H'')|] &\geq \mathbb{E}[|E(H')|-|B'|] = \mathbb{E}[|E(H')|] - \mathbb{E}[|B'|] \\
    &\geq(1-o(1))\left(\frac{|E(H)|}{4((2k-1)f)^2}\right)-\frac{|B|}{8((2k-1)f)^3} \\ 
    &\geq (1-o(1)) \left(\frac{m}{4((2k-1)f)^2}\right) - \frac{(2k-1)fm}{8((2k-1)f)^3} \\
    &\geq (1-o(1))\left(\frac{m}{4((2k-1)f)^2}\right)-\frac{m}{8((2k-1)f)^2} \\
    &= (1-o(1)) \left( \frac{m}{8((2k-1)f)^2} \right) \\
    &=\Omega\left(\frac{m}{(kf)^2}\right)
\end{align*}

Note that the  bounds on $|V(H'')|$ and on the girth of $H''$ are deterministic.  So there is some subgraph which has those bounds and where the number of edges is at least the expectation, proving the lemma.  
\end{proof}

This lemma allows us to prove the size bound.

\begin{theorem} \label{thm:size}
The subgraph $H$ returned by Algorithm~\ref{ALG:new} has at most $O\left(kf^{1-1/k} n^{1+1/k}\right)$ edges.
\end{theorem}
\begin{proof}
If $f = \Omega(n)$ then the theorem is trivially true.  Otherwise, by Lemmas \ref{lemma: block set size} and \ref{lemma: subgraph bound} we know that $H$ has a subgraph $S$ of girth larger than $2k$ on $O(n/(kf))$ nodes and with  $|E(S)| \geq \Omega\left(\frac{|E(H)|}{(kf)^2}\right)$ edges.  But it has long been known that any graph with $n$ vertices and girth larger than $2k$ must have at most $O(n^{1+1/k})$ edges (this is the key fact used in the original non-fault-tolerant greedy algorithm analysis~\cite{AlthoferDDJS:93}).  Hence $|E(S)| \leq O((n/(kf))^{1+1/k})$.  Therefore there are constants $c_1, c_2 > 0$ such that for large enough $n$,
\begin{align*}
    &c_1\left(\frac{n}{kf}\right)^{1+1/k} \geq |E(S)| \geq c_2 \left(\frac{|E(H)|}{(kf)^2}\right) \\
    \implies &|E(H)| \leq O\left((kf)^{1-1/k} n^{1+1/k}\right) = O\left(kf^{1-1/k}n^{1+1/k}\right). \qedhere
\end{align*}
\end{proof}

\begin{theorem} \label{thm:time}
The worst-case running time of Algorithm~\ref{ALG:new} is at most $O\left(m k  f^{2-1/k} n^{1+1/k}\right)$.
\end{theorem}
\begin{proof}
Algorithm~\ref{ALG:new} has $|E| = m$ iterations, each of which consists of one call to Algorithm~\ref{ALG:LBC} with $\alpha = f$ on graph $H$.  So the running time of each iteration (by Theorem~\ref{thm:LBC}) is at most $O((|E(H)| + n)f)$.  Theorem~\ref{thm:size} implies that $|E(H)| \leq O(kf^{1-1/k} n^{1+1/k})$, and thus the total running time is at most $O(m kf^{2-1/k} n^{1+1/k})$.
\end{proof}

Theorems~\ref{thm:correct-unweighted}, \ref{thm:size}, and \ref{thm:time} together imply Theorem~\ref{thm:main} in the unweighted case.

\section{Weighted Graphs}

We now show that we can use the algorithm we designed for the unweighted setting even in the presence of weights.  Our algorithm is very simple: we order the edges in nondecreasing weight order, but then run the \emph{unweighted} algorithm on the edges in this order.  We give this algorithm more formally as Algorithm~\ref{ALG:weight}.  Again, changing to edge fault-tolerance is straightforward: we just use the edge version of Algorithm~\ref{ALG:LBC}.  So we prove this only for vertex fault-tolerance for simplicity. 

\begin{algorithm}
\caption{Modified Greedy VFT Spanner Algorithm (Weighted)}
\label{ALG:weight}
\begin{algorithmic}
\STATE $\mathbf{function}$ FT-GREEDY$(G = (V,E,w),k,f)$\\
\STATE $H\leftarrow (V,\emptyset,w)$
\FORALL{$\{u,v\}\in E$ in nondecreasing weight order} 
    \IF{Algorithm~\ref{ALG:LBC} returns YES when run on input graph $H$ (with no weights) with terminals $u,v$ and $t = 2k-1$ and $\alpha = f$}
   \STATE Add $\{u,v\}$ to H
    \ENDIF
\ENDFOR
\RETURN H
\end{algorithmic}
\end{algorithm}

\begin{theorem} \label{thm:weights}
Algorithm~\ref{ALG:weight} returns an $f$-VFT $(2k-1)$-spanner with at most $O(kf^{1-1/k} n^{1+1/k})$ edges in time at most $O(m k  f^{2-1/k} n^{1+1/k})$.
\end{theorem}
\begin{proof}
The running time is directly from Theorem~\ref{thm:time}, since the only additional step in the algorithm is sorting the edges by weight, which takes only $O(m \log m)$ additional time.  The size also follows directly from Theorem~\ref{thm:size}, since Algorithm~\ref{ALG:weight} is just a particular instantiation of Algorithm~\ref{ALG:new} where the ordering (which is unspecified in Algorithm~\ref{ALG:new}) is determined by the weights.  In other words, Theorem~\ref{thm:size} holds for an \emph{arbitrary} order, so it certainly holds for the weight ordering.  

The more interesting part of this theorem is correctness: why does this algorithm return an $f$-VFT $(2k-1)$-spanner despite ignoring weights?  Let $F\subseteq V$ be an arbitrary fault set  with $|F|\leq f$ and $\{u,v\} \in E$ with $u,v \not\in F$ and $d_{G \setminus F}(u,v) = w(u,v)$.  By Lemma~\ref{lem:simple}, we just need to show that $d_{H \setminus F}(u,v) \leq (2k-1) w(u,v)$ in order to prove the theorem.  Clearly this is true if $\{u,v\} \in E(H)$.  So suppose that $\{u,v\} \not \in E(H)$.  Then when the algorithm considered $\{u,v\}$ it must have been the case that Algorithm~\ref{ALG:LBC} returned NO, and hence by Theorem~\ref{thm:LBC} every length-$(2k-1)$-cut in $H$ (unweighted) for $u,v$ has size larger than $f$ and so $F$ is not such a cut.  Thus at the time the algorithm was considering $\{u,v\}$, there was some path $P$ between $u$ and $v$ in $H \setminus F$ with at most $2k-1$ edges.  But since we considered edges in order of nondecreasing weight, every edge in $P$ has weight at most $w(u,v)$.  Thus 
\begin{align*}
d_{H \setminus F}(u,v) &\leq \sum_{e \in P} w(e) \leq \sum_{e \in P} w(u,v) = |P| w(u,v) \\
&\leq (2k-1) w(u,v),
\end{align*}
as required.
\end{proof}

\section{Distributed Algorithms}
In this section we give efficient randomized algorithms to compute fault-tolerant spanners of weighted graphs in two standard distributed models: the LOCAL model and the CONGEST model~\cite{peleg2000distributed}.  Recall that in both models we assume communication happens in synchronous rounds, and our goal is to minimize the number of rounds needed.  In the LOCAL model each node can send an arbitrary message on each incident edge in each round, while in the CONGEST model these messages must have size at most $O(\log n)$ bits (or $O(1)$ words, so we can send a constant number of node IDs and weights in each message).  Note that both models allow unlimited computation at each node, and hence the difficulty with applying the greedy algorithm is not the exponential running time, but its inherently sequential nature.

\subsection{LOCAL}
In the LOCAL model we will be able to implement the greedy algorithm at only a small extra cost in the size of the spanner.  Our approach is simple: we use standard network decompositions to decompose the graph into clusters, run the greedy algorithm in each cluster, and then take the union of the spanner for each cluster.  

The following theorem is a simple corollary of the construction of ``padded decompositions'' given explicitly in previous work on fault-tolerant spanners~\cite{DinitzK:11}.  It also appears implicitly in various forms in~\cite{LinialS93,Bartal96,miller2013,miller2015} (among others).  In what follows, the \emph{hop diameter} of a cluster refers to its \emph{unweighted} diameter.
\begin{theorem} \label{thm:decomposition}
There is an algorithm in the LOCAL model which runs in $O(\log n)$ rounds and constructs $P_1, P_2, \dots, P_{\ell}$ such that:
\begin{enumerate}
    \item Each $P_i$ is a partition of $V$, with each part of the partition referred to as a \emph{cluster}.  Let $\mathcal C = \cup_{i=1}^{\ell} P_i$ be the collection of all clusters of all $\ell$ partitions.  
    \item Each cluster has hop diameter at most $O(\log n)$ and contains some special node known as the \emph{cluster center}.
    \item $\ell = O(\log n)$ (there are $O(\log n)$ partitions).
    \item With high probability ($1 - 1/n^c$ for any constant $c$) for every edge $e \in E$ there is a cluster $C \in \mathcal C$ such that $e \subseteq C$.
\end{enumerate}
\end{theorem}

With this tool, it is easy to describe our algorithm.  First we use Theorem~\ref{thm:decomposition} to construct the partitions.  Then in each cluster $C$ we gather at the cluster center the entire subgraph $G[C]$ induced by that cluster.  Each cluster center uses the greedy algorithm (Algorithm~\ref{ALG:old}) on $G[C]$ to construct an $f$-VFT $(2k-1)$-spanner $H_C$ of $G[C]$, and then sends out the selected edges to the nodes in $C$.  Let $H$ be the final subgraph created (the union of the edges of each $H_C$)

\begin{theorem} \label{thm:local}
With high probability, $H$ is an $f$-VFT $(2k-1)$-spanner of $G$ with at most $O\left(f^{1-1/k} n^{1+1/k} \log n\right)$ edges and the algorithm terminates in $O(\log n)$ rounds.
\end{theorem}
\begin{proof}
The round complexity is obvious from the round complexity and cluster hop diameter bounds in Theorem~\ref{thm:decomposition}.

The total number of edges added is at most
\begin{align*}
    \sum_{i=1}^{\ell} \sum_{C \in P_i} |E(H_C)| &\leq \sum_{i=1}^{\ell} \sum_{C \in P_i} f^{1-1/k} |V(H_C)|^{1+1/k} \\
    &= f^{1-1/k} \sum_{i=1}^{\ell} \sum_{C \in P_i} |C|^{1+1/k} \\
    &\leq f^{1-1/k} \sum_{i=1}^{\ell} n^{1+1/k}\\
    &= O\left(f^{1-1/k} n^{1+1/k} \log n\right),
\end{align*}
where we used the size bound on the greedy algorithm from~\cite{BP19} and the fact from Theorem~\ref{thm:decomposition} that each $P_i$ is a partition of $V$.

To show correctness, consider some $\{u,v\} \in E$ and $F \subseteq V$ with $|F| \leq f$ and $u,v\not\in F$ so that $d_{G \setminus F}(u,v) = w(u,v)$.  By Lemma~\ref{lem:simple}, we just need to prove that $d_{H \setminus F}(u,v) \leq (2k-1) w(u,v)$.  Let $C \in \mathcal C$ be a cluster which contains both $u$ and $v$, which we know exists (with high probability) from Theorem~\ref{thm:decomposition}. Let $F_C = F \cap C$.  Then
\begin{align*}
    d_{H \setminus F}(u,v) &\leq d_{H_C \setminus F_C}(u,v) \\
    & \leq (2k-1) \cdot d_{G[C] \setminus F_C}(u,v) & \text{(definition of $H_C$)} \\
    & \leq (2k-1) \cdot w(u,v) & (\{u,v\} \in E(G[C] \setminus F_C) )
\end{align*}
Thus $H$ is indeed an $f$-VFT $(2k-1)$-spanner of $G$.  
\end{proof}

\subsection{CONGEST}
We unfortunately cannot use the approach that we used in the LOCAL model in the CONGEST model, since we cannot efficiently gather the entire topology of a cluster at a single node.  We will instead use the fault-tolerant spanner of Dinitz and Krauthgamer~\cite{DinitzK:11}, rather than the greedy algorithm, and combine it with the non-fault-tolerant spanner of~\cite{BaswanaS:07} which can be efficiently constructed in CONGEST.  This approach means that, unlike in the centralized setting or the LOCAL model, we will not be able to get size-optimal fault-tolerant spanners.

The algorithm of~\cite{DinitzK:11} works as follows (in the traditional centralized model).  Suppose that we have some algorithm $\mathcal A$ which constructs a $(2k-1)$-spanner with at most $g(n)$ edges on any graph with $n$ nodes.  The algorithm of~\cite{DinitzK:11} consists of $O(f^3 \log n)$ iterations, and in each iteration every node chooses to participate independently with probability $1/f$.  For each $i \in O(f^3 \log n)$, let $V_i$ be the vertices who participate and let $G_i$ be the subgraph of $G$ induced by them.  We let $H_i$ be the $(2k-1)$-spanner constructed by $\mathcal A$ on $G_i$.  Then we return the union of all $H_i$.

The main theorem that~\cite{DinitzK:11} proved about this is the following.

\begin{theorem}[\cite{DinitzK:11}] \label{thm:DK}
This algorithm returns an $f$-VFT $(2k-1)$-spanner of $G$ with $O\left(f^3 g((2n)/f) \log n\right)$ edges with high probability. 
\end{theorem}

Note that when $g(n) = n^{1+1/k}$, this results in an $f$-VFT $(2k-1)$-spanner with at most $O(f^{2-1/k} n^{1+1/k} \log n)$, which is precisely the bound from~\cite{DinitzK:11}.

Since the algorithm of~\cite{DinitzK:11} uses an \emph{arbitrary} non-fault-tolerant spanner algorithm $\mathcal A$, by using a \emph{distributed} spanner algorithm for $\mathcal A$ we naturally end up with a distributed fault-tolerant spanner algorithm.  In particular, we will combine the algorithm of~\cite{DinitzK:11} with the following algorithm due to Baswana and Sen~\cite{BaswanaS:07}.

\begin{theorem}[\cite{BaswanaS:07}] \label{thm:BS}
There is an algorithm that computes a $(2k-1)$-spanner with at most $O(k n^{1+1/k})$ edges of any weighted graph in $O(k^2)$ rounds in the CONGEST model.
\end{theorem}

Combining Theorems~\ref{thm:DK} and \ref{thm:BS} immediately gives an algorithm in CONGEST that returns an $f$-VFT $(2k-1)$-spanner of size at most $O(k f^{2-1/k} n^{1+1/k})$ that runs in at most $O(k^2 f^3 \log n)$ rounds (with high probability).  We can just run each iteration of the Dinitz-Krauthgamer algorithm~\cite{DinitzK:11} in series, and in each iteration we use the Baswana-Sen algorithm~\cite{BaswanaS:07}.  Since there are $O(f^3 \log n)$ iterations, and Baswana-Sen takes $O(k^2)$ rounds, this gives a total round complexity of $O(k^2 f^3 \log n)$.  

We can improve on this bound by taking advantage of the fact that each iteration of Dinitz-Krauthgamer runs on a relatively small graph (approximately $n/f$ nodes), so we can run some of these iterations in parallel.
\begin{theorem} \label{thm:congest}
There is an algorithm that computes an $f$-VFT $(2k-1)$-spanner of $G$ with $O\left(k f^{2-1/k} n^{1+1/k} \log n\right)$ edges of any weighted graph and which runs in $O(f^2(\log f + \log \log n) + k^2 f \log n)$ rounds in the CONGEST model (all with high probability).
\end{theorem}
\begin{proof}
In the first phase of the algorithm each vertex randomly selects which of the $O(f^3 \log n)$ iterations in which to participate by choosing each iteration independently with probability $1/f$.  So by a Chernoff bound, with high probability every node picks $O(f^2 \log n)$ iterations in which to participate.  Then each vertex sends its chosen iterations to all of its neighbors.  Identifying these iterations take $O(f^2 \log n \cdot \log(f^3 \log n)) = O(f^2 \log n \cdot (\log f + \log\log n))$ bits, and thus $O(f^2 (\log f + \log\log n))$ rounds in CONGEST.  

After this has completed we enter the second phase of the algorithm, and now every node knows which iterations it is participating in and which iterations each of its neighbors is participating in.  With high probability (by a simple Chernoff bound), for every edge there are at most $O(f \log n)$ iterations in which \emph{both} endpoints participate.  Thus if we try to run all $O(f^3 \log n)$ iterations of Baswana-Sen (Theorem~\ref{thm:BS}) in parallel, we have ``congestion'' of $O(f \log n)$ on each edge (at each time step) since there could be up to that many iterations in which a message is supposed to be sent along that edge at that time.  Thus we can simply use $O(f \log n)$ time steps for each time step of Baswana-Sen and can simulate all $O(f^3 \log n)$ iterations of the Dinitz-Krauthgamer algorithm (note that each Baswana-Sen message needs to have a tag added to it with the iteration number, but since that takes at most $O(\log(f^3 \log n)) = O(\log f + \log\log n) \leq O(\log n)$ bits it fits within the required message size).  Hence the total running time of this second phase is at most $O(k^2 f \log n)$.

The size and correctness bounds are direct from Theorems~\ref{thm:DK} and \ref{thm:BS}, and the round complexity is from our analysis of the two phases above.
\end{proof}

\section{Conclusion and Future Work}
In this paper we designed an algorithm to compute nearly-optimal fault-tolerant spanners in polynomial time, answering a question posed by~\cite{BDPW18,BP19}.  We also gave an optimal construction in the LOCAL model which runs in $O(\log n)$ rounds, and an efficient algorithm in the CONGEST model that constructs fault-tolerant spanners which have the same size as in~\cite{DinitzK:11} rather than the optimal size.

There are many interesting open questions remaining about efficient algorithms for fault-tolerant spanners, as well as about the extremal properties of these spanners.  Most obviously, the size we achieve is a factor of $k$ away from the optimal size, due to our use of an $O(k)$-approximation for \textsc{Length-Bounded Cut}.  Can this be removed, either by giving a better approximation for \textsc{Length-Bounded Cut} or through some other construction?  While $k$ is somewhat small since spanners tend to be most useful for constant stretch (and never have stretch larger than $O(\log n)$), it would still be nice to get fully optimal size in polynomial time.  Similarly, our distributed constructions are extremely simple, and there is no reason to think that we actually need $\Omega(\log n)$ rounds in LOCAL or that we cannot get optimal size fault-tolerant spanners in CONGEST.  It would be interesting to design better distributed and parallel algorithms for these objects, particularly since the greedy algorithm (the only size-optimal algorithm we know) tends to be difficult to parallelize. 

From a structural point of view, we reiterate one of the main open questions from~\cite{BDPW18} and~\cite{BP19}: understanding the optimal bounds for \emph{edge}-fault-tolerant spanners.  The best upper bound we have is the same $O(f^{1-1/k} n^{1+1/k})$ that we have for the vertex case, while the best lower bound is $\Omega(f^{\frac12 (1-1/k)} n^{1+1/k})$ (from~\cite{BDPW18}).  What is the correct bound?

\ifarxiv
\bibliographystyle{alpha}
\fi
\bibliography{references.bib}
\end{document}